\documentclass{article}
\usepackage{graphicx} 
\usepackage{amsmath}
\usepackage{amsfonts}
\usepackage{amsthm}
\usepackage{wasysym}
\usepackage{comment}
\usepackage{authblk}

\theoremstyle{definition}
\newtheorem{theorem}{Theorem}
\newtheorem{proposition}{Proposition}
\newtheorem{lemma}{Lemma}
\newtheorem{remark}{Remark}
\newtheorem{corollary}{Corollary}

\newtheorem{observation}{Observation}

\usepackage{fullpage}
\usepackage[colorlinks=true, urlcolor=blue, linkcolor=red]{hyperref}

\usepackage[]{color-edits}
\addauthor{EF}{blue}

\addauthor{AM}{red}

\addauthor{BB}{cyan}
\newcommand{\bbc}[1]{{\BBcomment{#1}}}

\newcommand{\AThr}{\mathcal{A}_{\mathcal{THR}}}
\newcommand{\A}{\mathcal{A}}
\newcommand{\piSuff}{\pi_{\mathrm{suff}}}

\title{Economic Censorship Games in Fraud Proofs}
\author[a]{Ben Berger}
\author[a]{Edward W. Felten}
\author[a]{Akaki Mamageishvili}
\author[b]{Benny Sudakov}

\affil[a]{Offchain Labs}
\affil[b]{ETH Zurich}
\date{}

\begin{document}

\maketitle

\begin{abstract}
    {
    Optimistic rollups rely on fraud proofs --- interactive protocols executed on Ethereum to resolve conflicting claims about the rollup's state --- to scale Ethereum securely.

    To mitigate against potential censorship of protocol moves, fraud proofs grant participants a significant time window, known as the challenge period, to ensure their  moves are processed on-chain. Major optimistic rollups today set this period at roughly one week, mainly to guard against strong censorship that undermines Ethereum's own crypto-economic security.
    However, other forms of censorship are possible, 
    and their implication on optimistic rollup security is not well understood.
    
    This paper considers economic censorship attacks, where an attacker censors the defender’s transactions by bribing block proposers. 
    At each step, the attacker can either censor the defender --- depleting the defender’s time allowance at the cost of the bribe --- or allow the current transaction through while conserving funds for future censorship.

    We analyze three game-theoretic models of these dynamics and 
    determine the challenge period length required to ensure the defender’s success, as a function of the number of required protocol moves and the players’ available budgets.
    }
\end{abstract}


\section{Introduction}
Optimistic rollups are layer 2 scaling solutions designed to increase the throughput (processed transactions per second) of the Ethereum blockchain while reducing its costs 
{and inheriting its economic security}.
{At least from an economic value perspective, optimistic rollups are without a doubt the leading scaling solution for Ethereum, with the three leading examples Arbitrum, Base and Optimism securing over \$30B as of February 7th, 2025.\footnote{See \url{l2beat.com} for up-to-date numbers.}}

The main idea of an optimistic rollup is to execute incoming transactions outside the Ethereum main chain (layer 1), 
thereby reducing
its
computational load
and
enhancing its scalability.
To inherit Ethereum's security, 
optimistic rollups implement mechanisms to have their updated state (resulting from executing these transactions) posted and confirmed on Ethereum, at some pre-determined rate.\footnote{In the case of Arbitrum, this happens roughly once per hour.}

These mechanisms, usually referred to as \emph{fraud-proofs} or \emph{dispute resolution protocols}, work as follows:
one party makes a claim regarding the updated state of the rollup chain,
by sending a corresponding transaction to a designated smart contract on Ethereum.  The claim is accompanied by a deposit which is locked in escrow. 
This claim is not automatically accepted. Instead, it opens the door for any other party, whether white-listed or not, to challenge its validity by posting a counter claim.

The ensuing challenge protocol is an interactive game where the honest party must make correct moves to advance the state of the game towards resolution. Each such move is a transaction sent by that party to one of the smart contracts that govern the protocol on Ethereum.  Once the game terminates, one of the claims is confirmed, the winning party gets her deposit back and the losing parties lose their deposits.
Under ideal conditions, the honest party is guaranteed to win, ensuring the confirmed state's integrity.

However, this system is not without its vulnerabilities, one of which is the potential for a \emph{censorship attack} at the base layer.
Censorship occurs when a transaction's inclusion in the canonical Ethereum chain is prevented or delayed due to malicious activity.

In the context of the challenge protocol, censorship of honest moves can prevent
the challenge from being resolved in favor of the honest party.
To mitigate this risk, existing dispute resolution protocols provide a significant time window for participating parties to submit their required moves before getting eliminated.

A natural question is how big this time window --- usually referred to as the \emph{challenge period} --- should be.
On the one hand, the challenge period should be set high enough to attain overwhelming confidence that the honest claim gets confirmed.
On the other hand, the challenge period is a lower bound on the confirmation time of the challenge protocol.  
Thus, the lower it can be set without compromising safety, the better.

In practice, the challenge period in all existing dispute resolution protocols is designed with a particular type of censorship in mind, usually referred to as ``strong'' censorship ---
an adversary can strongly censor when it has
enough control over Ethereum's stake to allow for complete jurisdiction over which blocks get accepted as part of the canonical chain
\cite{buterin2020}.

It is widely accepted in the Ethereum community that such an attack cannot occur for more than 7 days without a proper social response, which involves a hard fork of the chain.
For this reason, the major optimistic rollups have set their challenge periods to roughly 7 days.\footnote{See a recent discussion on the topic here:\\ \text{https://ethresear.ch/t/optimistic-rollups-the-challenge-period-and-strong-censorship-attacks/21721}}


However, 
other forms of censorship are possible and may be more practical to accomplish.
To the best of our knowledge, there has not been a scientific line of inquiry aimed to study these and to provide an informed answer to the question above.

\subsection{Our Contributions}

In this paper, we consider \emph{economic censorship attacks}
which exploit the fact that in Ethereum,
block proposers/builders have absolute control over which transactions are included in their block.
In such an attack, the adversary
bribes the proposers to exclude the honest party's moves in the challenge protocol, by offering them a payment higher than the honest party's inclusion tip.

The challenge period serves a dual purpose here: it not only provides time for the honest party to react, but also leverages the economic dynamics of the attack. Since the adversary must continuously censor transactions across nearly all blocks during this period
while the defender can afford to make claims sporadically, 
the defender gains an economic advantage. Over time, this advantage can lead to resource depletion for the attacker, making sustained censorship less viable.

Our goal is to understand these dynamics, the strategies we expect the different actors to use, and the implications of these for optimistic rollup security and design.
In particular, we aim not only to highlight vulnerabilities, but to also suggest efficient and simple defense strategies in the fraud proof mechanism, ensuring that optimistic rollups can maintain integrity and resist sophisticated forms of economic attack.





To this end, we study a stylized model of the interaction between the attacker, the defender and the block proposers, in the form of a multi-round strategic game.
Each round corresponds to a block where the defender wants to submit its next move of the underlying challenge protocol.
We study three variants of this game, which differ according the number of proposers responsible for deciding the current block's contents, and the proposers' assumed behavior.
We take a pessimistic view of the game from the defense perspective. In this view, the attacker can always observe the defender's action, namely, observe the tip offered by the defender to the proposers for inclusion and decide on her response after that. 

The three variants are parameterized at least by the total number of rounds (corresponding to the length of the challenge period), 
the number of those rounds that the defender needs to win (corresponding to the number of honest moves required to win the fraud proof challenge), and the budgets of the attacker and the defender.
The main question we ask is:

\begin{quote}
    How large does the attacker's (defender's) budget have to be with respect to the defender's (attacker's) budget, as a function of the game's parameters, such that it can guarantee victory in the game?
\end{quote}

We provide answers to this question for each of the variants.  En route, we identify optimal and approximately optimal (but simpler) strategies for both players.
%
In what follows we describe our results for each of the variants in more detail.

\paragraph{$\mathcal{G}^1$ Game.}

First, we look into the most basic version of this interaction where there is a single proposer in each round, and the round is won by the party that has placed a higher bid.  Thus the proposer is assumed to be myopic, which is reasonable since in practice, the proposer for each block in Ethereum is chosen randomly from the pool of validators. The round winner gets her budget deducted by an amount equal to her bid, and the defender wins the entire game if and only if it has won the required number of rounds.

This foundational model allows us to gain initial insights into the strategic dynamics at play, serving as a crucial building block for understanding more complex interactions.

We identify the players' optimal strategies, and we 
find a linear inequality between the attacker and defender budgets which holds (doesn't hold) if and only if the attacker (defender) can guarantee victory
--- see Theorem \ref{thm:one_builder}.

\paragraph{$\mathcal{G}^1_k$ Game.}
Next, we investigate an extension of $\mathcal{G}^1$ where a round can be either regular or \emph{special}. In a special round the attacker is required to outbid the defender by a factor  $k>1$ in order to win it.
These special rounds are motivated by the existence of Ethereum proposers that are not rational, but rather follow the default transaction inclusion policy offered by Ethereum's execution clients.
Under this policy, transactions are included in the block by decreasing order of \emph{priority fee per gas}, so long as the block gas limit is respected.\footnote{Gas is Ethereum's universal unit for computational load measurement, and each transaction consumes some amount of gas.  Today, Ethereum blocks cannot contain more than 30M gas worth of transactions.}

To illustrate the implication of this for censorship, let us assume that the honest party's transaction consumes $g$ gas units and is specified with a priority fee per gas $f$ which happens to be the highest one among the current pending transactions. 
An adversary who wishes to censor that transaction would need to spend at least $(G-g)f$ in priority fees, where $G$ is the block gas limit, by submitting at least $G-g$ gas units worth of transactions with priority fee at least $f$.  This is $k=(G-g)/g$ times as much as the honest party would have paid for inclusion.


In this extended game, the economic cost of censorship increases as the number of special rounds increases, potentially deterring attackers further while allowing defenders to employ more sophisticated bidding strategies.  We analyze how this change impacts the defender's strategy, the attacker's strategy, and the overall balance of the game, providing insights into how the existence of such rounds could be integrated into real-world fraud-proof bidding strategies to enhance security against economically driven attacks.

First, we assume that
the exact number and location of the special rounds is known in advance to both players. 
We again identify the players' optimal strategies and a linear relation between their budgets that determines the player who can guarantee victory --- see Theorem \ref{thm:optimal_specials}.

As opposed to $\mathcal{G}^1$, the optimal strategies in this variant are complex and the best bid in each round changes dramatically as the game progresses.  
However, we identify simple strategies for both players whose performance converges to that of the optimal ones in some asymptotic regime of the game's parameters --- see Theorem \ref{optimal_vs_trivials} and Corollary \ref{cor:asymptotic}.  We remark that in practice, the relevant parameters are indeed situated in this regime --- see section \ref{sec:eval} for details.

These results have immediate implications for the randomized version of this game, in which it is assumed that each round is special with some known probability, that the rounds are independent, and that the parties learn if some round is special or regular only when it starts.  This randomized variant is particularly relevant in practice,
as it is estimated that around 2\% of Ethereum's proposers use the default transaction inclusion policy.
We show that playing the simple strategies from the first version is approximately optimal here as well --- see Proposition \ref{prop:probabilistic_variant}.

\paragraph{$\mathcal{G}^m$ Game.}
Finally, we consider a setting where a block's content is determined by multiple proposers, 
and the defender and attacker can interact with any of them.
The defender wins a round if at least one of the proposers includes the defender's transaction in its suggested ``inclusion list''.
A variation of this mechanism named FOCIL \cite{focil} is seriously considered for deployment in one of the upcoming Ethereum hard forks.

In this setting the game becomes more intricate since a proposer's payoff might depend on the behavior of other proposers, and this introduces uncertainty in the inclusion outcome.
This
on its own affects censorship and counter-censorship strategies. 
We examine 
how the competition among proposers and its resulting equilibrium behavior affects the dynamics and outcome of the game.
We obtain almost matching conditions for both players to have a winning strategy. 
In particular, for the attacker, we obtain a lower bound on its budget that is almost the number of proposers times higher than the value from the 
game $\mathcal{G}^1$,
for it to win with high probability --- see Theorems \ref{condition_alice_k} and \ref{condition_daria_k}, and Proposition \ref{condition_daria_k_assymetric}.

We conclude the paper by evaluating our theoretical results under real world parameters (section \ref{sec:eval}), followed by a conclusion where we present ideas for future work (section \ref{sec:conclusion}).

\subsection{Related Literature}

Challenge protocols in the context of optimistic rollups are described in~\cite{arbitrum_classic} and~\cite{bold}. 
Verified computation, that pre-dates fraud-proofs, has been studied in~\cite{short_games} and~\cite{verified_computation}.
~\cite{censorship_resistance} studies a similar economic censorship model to ours in the private-value auction context.

Multiple proposers have received attention by the Ethereum community~\cite{mcp}.
In another live Ethereum Improvement Proposal named FOCIL~\cite{focil}
there is a committee of validators for each of block that propose inclusion lists, in addition to the main block proposer. 
While the actual specification allows the main builder to prioritize the whole block with its own transactions and not include the committee suggested transactions if the block is already full, such a scenario can only happen rarely
due to Ethereum's gas pricing mechanism, where consistently full blocks would lead to an exponentially escalating gas price.
~\cite{aucil} study an auction based mechanism for aggregating inclusion lists from committee members. 

\section{Preliminaries}

We model two-player games between a
defender Daria and 
an
attacker Alice. Daria has a budget $D$ to play the game and Alice has a budget $A$. There are $N$ transactions that Daria needs to publish. The games are played in rounds, and there are $T$ rounds in total.  These correspond to the number of blocks during the delay period, as described in the introduction. {At most one of Daria's $N$ transactions can be included in a single round/block.}

In the following sections, we describe three versions of the game. In the first game, denoted by $\mathcal{G}^1$, a single block builder determines whether to include or censor Daria's transaction based on which of the two players offers the higher bid. 

The second game, denoted $\mathcal{G}^1_k$ for some parameter $k\geq 1$, is similar to $\mathcal{G}^1$ --- the only difference being that some of the rounds in this variant are ``special''. In a special round, Alice needs to out-bid Daria by a factor of $k$ in order to win the round (hence, $\mathcal{G}^1$ and $\mathcal{G}^1_1$ are identical).

In the third game, denoted by $\mathcal{G}^m$, there are $m$ block builders or members of the inclusion committee in each round that participate in the block building process.
{Daria's transaction is included in the block if at least one of the builders or committee members suggests to include it.}

{In all these games, the builders are assumed to be myopic, in the sense that they only care about maximizing their pay-off in the current round.  This is a reasonable assumption since the identity of the block builders can change between rounds in practice.}
{$\mathcal{G}^1$, $\mathcal{G}^1_k$ and $\mathcal{G}^m$ are formally described and analyzed in sections \ref{sec:g_1_game},~\ref{sec:g_1_k_game} and~\ref{sec:g_m_game}, respectively.}


\section{$\mathcal{G}^1$ Game}
\label{sec:g_1_game}



Consider the game $\mathcal{G}^1$ between Alice and Daria, which is conducted in rounds, with the following parameters:

\begin{itemize}
\item $T$ - total number of rounds (assume $T \geq 1$)
\item $N$ - number of rounds that Daria has to win to be victorious (assume $1 \leq N \leq T$)
\item $A$ - Alice’s budget (assume $0 \leq A$) 
\item $D$ - Daria’s budget (assume $0 \leq D$).
\end{itemize}

A round is conducted as follows:

\begin{enumerate}
\item  Daria offers a bid $b$, which is at most her current budget.
\item  Alice observes $b$ and then offers a counter-bid $b'$, which is at most her current budget.
\item 
    Alice wins the round iff $b'\geq b$, and otherwise Daria wins the round.
\item The round winner gets her corresponding budget deducted by an amount equal to her bid, and a new round begins.
\end{enumerate}

The game ends when Daria wins $N$ rounds or Alice wins $T-(N-1)$ rounds (note that only one of these can and must happen), and the corresponding party wins the entire game.

Note 
that this game is a finite two-person game of perfect information in which the players move alternately and in which chance does not affect the decision-making process.  Therefore, by Zermelo’s Theorem, one of the players can force victory.

Our goal here is to understand which of the two players has a winning strategy, as a function of the parameters of the game.  
%
In the following theorem, we fully characterize the outcome of $\mathcal{G}^1$.

\begin{theorem}\label{thm:one_builder}
Alice has a winning strategy in $\mathcal{G}^1$ with parameters $T,N,D,A$ iff
\[
\frac{T-N+1}{N}\cdot D \leq A.
\]
\end{theorem} 

\begin{proof}
We first prove the (counter-positive of the) "only if" direction: if $T-N+1>\frac{AN}{D}$, then Daria has a winning strategy. Consider the following strategy by Daria: in each round $i$ bid $b_i=\frac{D}{N}$. The cheapest strategy for Alice against this strategy is to let Daria win $N-1$ rounds and win the remaining $T-N+1$ rounds. Winning the round for Alice means spending at least $\frac{D}{N}$, that is, in total Alice spends $(T-N+1)\frac{D}{N}$, which, by the assumption is larger than Alice's budget $A$. Therefore, this is a winning strategy for Daria.   

For the other direction, assume towards contradiction that Daria has a winning strategy but $T-N+1 \leq \frac{AN}{D}$. Then, consider the following strategy by Alice: she lets Daria win every round in which Daria bids more than $\frac{D}{N}$ and does not let her win any round in which Daria bids less or equal that $\frac{D}{N}$. We show that with this strategy it is, in fact, Alice who wins. Daria cannot enforce the win of $N$ rounds herself as she needs to pay strictly more than $N\frac{D}{N}=D$ and she cannot exhaust Alice's budget by making her pay for $T-(N-1)$ rounds since by our assumption $(T-N+1)\frac{D}{N} \leq A$.  We have reached a contradiction and this concludes the proof of the theorem.
We remark that  the "only if" direction can also be proved exactly like the "if" direction by switching the roles of Alice and Daria, and having Alice bid $\frac{A}{T-N+1}$ in every round.
\end{proof}

\section{$\mathcal{G}^1_k$ Game}
\label{sec:g_1_k_game}

Consider the game $\mathcal{G}^1_k$ between Alice and Daria, which extends the game $\mathcal{G}^1$ (as discussed in Section \ref{sec:g_1_game}) in the following way.  First, it is parametrized by $k\geq 1$, called the \emph{special round factor}.  This parameter stays constant throughout the game execution, as opposed to the rest of the parameters.
In addition to the parameters $T,N,D,A$ as introduced in the previous section, the game $\mathcal{G}^1_k$ also has the following parameter:
\begin{itemize}
    \item A binary string $\Pi=(P_1, \ldots, P_{T}) \in \{\mathbb{S},\mathbb{R}\}^{T}$, specifying a special and regular round configuration.  $\mathbb{S}$ and $\mathbb{R}$ denote a special and regular (or non-special) round, respectively.
\end{itemize}
Both players are assumed to know exactly which rounds are special.  In other words, both players know the string $\Pi$. 

The game is conducted in rounds as before.  If a round is regular (or non-special), it is conducted and decided exactly as in $\mathcal{G}^1$. A special round is conducted like a regular round, but the round winner is decided differently.  Here, Alice wins the round iff $b' \geq k\cdot b$, where $b'$ is Alice's bid and $b$ is Daria's bid.

In what follows we analyze $\mathcal{G}^1_k$.

\subsection{Analysis}
We start by introducing some useful notation.
The \emph{state} of the game prior to each round is characterized by the tuple 
$(t,n,\pi,d,a)$,
where $t$ is the number of rounds remaining in the game, $n$ is the number of remaining rounds Daria needs to win, 
$\pi \in \{\mathbb{S},\mathbb{R}\}^t$ indicates the remaining special and regular round configuration,
$d$ is Daria’s remaining budget and $a$ is Alice’s remaining budget (the parameter $k$ is fixed throughout, so we omit it from the description of the state).  The state at the beginning of the game is therefore captured by 
$(T,N,\Pi, D,A)$.

We can describe the state transition of the game as follows.  Given state 
$(t,n,\pi,d,a)$, where $1\leq n \leq t$ and $\pi = (p_1, \ldots, p_t) \in \{\mathbb{S},\mathbb{R}\}^t$: 

\begin{enumerate}
\item Daria offers bid $b\leq d$
\item Alice sees $b$ and offers a counter-bid $b' \leq a$.
\item Let $\piSuff$ be the suffix $ \piSuff= (p_2, \ldots, p_t)$.  If $p_1 = \mathbb{R}$:
    \begin{enumerate}
    \item If $b'\geq b$, go to state 
    $(t-1, n, \piSuff, d, a-b')$.
    \item Otherwise, go to state 
    $(t-1, n-1, \piSuff, d-b, a)$
    \end{enumerate}
\item Else ($p_1 = \mathbb{S}$):
    \begin{enumerate}
    \item If $b'\geq k\cdot b$, go to state $(t-1,n,
    \piSuff,
    d,a-b')$
    \item Otherwise, go to state $(t-1,n-1,
    \piSuff,
    d-b,a)$
    \end{enumerate}
\end{enumerate}

Recall that by Zermelo's theorem, one of the players has a winning strategy given any starting state.  Given the state $(t,n,\pi,d,a)$, we denote by $W(t,n,\pi,d,a) \in \{Alice, Daria\}$ the player that has a winning strategy starting from that state.
Note that if $W(t,n,\pi,d,a) = Alice$, then $W(t,n,\pi,d,a') = Alice$ for any alternative budget $a'>a$ --- Alice can use exactly the same strategy she would have used with the budget $a$.
We also note that for any $t,n,\pi,d$, there is some $a$ such that $W(t,n,\pi,d,a) = Alice$ --- e.g. with budget $a = t\cdot k \cdot d$, Alice wins by bidding $b' = k\cdot d$ in every round until victory.
We thus define, for any $t,n,\pi,d$:
\[
\AThr(t,n,\pi,d) := \min\{a \mid W(t,n,\pi,d,a) = Alice\}
\]

Our goal here is to understand the function $\AThr$, which we refer to as \emph{Alice's winning threshold}.

\begin{theorem}\label{thm:optimal_specials}
Alice's winning threshold $\AThr$ is linear in Daria's budget $d$, and it is indifferent to the internal ordering of the configuration $\pi$, i.e., it depends only on the number of special rounds (equivalently, regular rounds) in $\pi$.  Formally, there exists a function $\A$ such that for any $1\leq t,1\leq n\leq t,\pi, 0 \leq d$, we have
\[
\AThr(t,n,\pi,d) = \A(t,n,s) \cdot d,
\]
where $s = \left|\{i \mid p_i = \mathbb{S}\}\right|$ is the number of special rounds in $\pi$.
Furthermore, $\A$ is determined by the following recurrence relation:
\begin{enumerate}
\item $\A(t,n,s) = \A(t-1,n-1,s-1)\cdot \left[\frac{k+\A(t-1,n,s-1)}{k+\A(t-1,n-1,s-1)} \right]$
\item $\A(t,n,s) = \A(t-1,n-1,s)\cdot \left[\frac{1+\A(t-1,n,s)}{1+\A(t-1,n-1,s)}\right]$
\item Boundary conditions:  
    \begin{enumerate}
        \item $\A(t,1,s)  = ks + (t-s)$
        \item $\A(t,n,0) = \frac{t-n+1}{n}$
        \item $\A(t,n,t) = k\left[\frac{t-n+1}{n}\right]$
        \item $\A(t,t,s) = \frac{k}{s+k(t-s)}$
    \end{enumerate}
\end{enumerate}
\end{theorem}

\begin{proof}
    Let $1 \leq t, 1\leq n\leq t, \pi \in \{\mathbb{S}, \mathbb{R}\}^t, 0 \leq d$.  Denote $s = \left|\{i \mid p_i = \mathbb{S}\}\right|$.
    We start with the boundary conditions.  
    For part 3-(a), we need to show that if $n=1$, then
    \[
    \AThr(t,n,\pi, d) = \left[ ks+(t-s) \right]\cdot d.
    \]
    If $n=1$, then Daria needs to win only one round to be victorious.  In that case she might as well bid her entire budget $d$ in every round, as there is no point in saving part of the budget for future rounds.  Alice wins the game iff she wins all rounds.  This can be achieved iff her budget is enough to win all special and regular rounds, given that Daria indeed always bids $d$. This holds iff $a \geq \left[ks+(t-s)\right]\cdot d$, and we are done with part 3-(a).  
    
    As for 3-(b), this is exactly the statement of Theorem~\ref{thm:one_builder}. 3-(c) is obtained in exactly the same way as Theorem~\ref{thm:one_builder}, by scaling the factor by which Alice needs to overbid Daria in the proof of that theorem by $k$.

    As for part 3-(d), if $n=t$ then Alice needs to win only one round to be victorious, and thus she might as well decide to win any round in which her budget allows for it.  Consequently, Daria's budget must be enough to allow her to bid strictly more than $a$ in all regular rounds and $\frac{a}{k}$ in all special rounds, where $a$ is Alice's budget.  We conclude that Daria can secure victory iff $d > \frac{a}{k}\cdot s + a\cdot (t-s)$.  Part 3-(d) is implied by negating and rearranging this inequality.

    We now proceed to prove that $\AThr$ is linear in $d$.  Specifically, we prove that there is a function $\AThr'$ such that 
    \[
    \AThr(t,n,\pi,d) = \AThr'(t,n,\pi) \cdot d.
    \]
    We prove this by induction on $t$.  For the base case, note that if $t=1$, then we must also have $n=1$ and we obtain the claim by part 3-(a) which we have already proved above.
    
    We now assume that $2\leq t$. 
    We can also assume that $2 \leq n < t$, as the claim has already been proved for the cases $n=1$ and $n=t$.
    Let $a$ be Alice's budget, and consider the round to be played at state $(t,n,\pi,d,a)$.
    Let $\piSuff=(p_2,\ldots,p_t)$ be the suffix of $\pi$, such that $\pi = (p_1, \piSuff)$.
    
    Let us also assume first that the round is special, i.e. $p_1 = \mathbb{S}$, and that Daria bids $b$.  Alice, having observed $b$, has two options:
    \begin{enumerate}
        \item Let Daria win the round, i.e. move the game to the state 
        $(t-1,n-1,
        \piSuff,
        d-b,a)$.
        By induction, Alice can force victory from here iff 
        \begin{align}
            a &\geq 
            \AThr'(t-1,n-1,\piSuff)(d-b).
            \label{eq:let_Daria_win}
        \end{align}
        \item If $a>kb$, then she can decide to win the round, i.e. move the game to the state $(t-1,n,
        \piSuff,
        d,a-kb)$.  By induction, Alice can force victory from here iff 
        \begin{align}
        \label{eq:let_Daria_lose}
            a &\geq 
            kb + 
            \AThr'(t-1,n,\piSuff) d
        \end{align}
    \end{enumerate}

    Of the two, Alice might as well choose the option that puts the least strain on her budget 
    according to Inequalities \ref{eq:let_Daria_win} and \ref{eq:let_Daria_lose}.
    Note that the right-hand side of Inequality \ref{eq:let_Daria_win} is a downward-sloping line as a function of the bid $b$, whereas the right-hand side of Inequality \ref{eq:let_Daria_lose} is an upward-sloping line as a function of $b$. 
    
    Let $b^*_{\mathbb{S}}$ be the bid that equalizes these two functions.  
    If $b>b^*_{\mathbb{S}}$, then Alice chooses the top option. 
    If $b<b^*_{\mathbb{S}}$, then she choose the lower option.  
    In either case, the winning condition is worse for Daria compared to bidding $b^*_{\mathbb{S}}$.
    
    Thus, Daria's best move is to bid $b^*_{\mathbb{S}}$,
    and by plugging that into (say) Inequality \ref{eq:let_Daria_win}, we get the desired winning condition for Alice.

    To this end we compute $b^*_{\mathbb{S}}$ by equating the right-hand sides of Inequalities \ref{eq:let_Daria_win} and \ref{eq:let_Daria_lose}:
    \begin{align}
        kb^*_{\mathbb{S}}+\AThr'(t-1,n,\piSuff)d      &= \AThr'(t-1,n-1,\piSuff)(d-b^*_{\mathbb{S}}) &\iff \nonumber \\
        b^*_{\mathbb{S}}(k+\AThr'(t-1,n-1,\piSuff))   &= d(\AThr'(t-1,n-1,\piSuff) - \AThr'(t-1,n,\piSuff)) &\iff \nonumber \\
        b^*_{\mathbb{S}}  &= \frac{\AThr'(t-1,n-1,\piSuff) - \AThr'(t-1,n,\piSuff)}{k+\AThr'(t-1,n-1,\piSuff)} d \label{eq:b^*_S}
    \end{align}

    Plugging the obtained $b^*_{\mathbb{S}}$ into Inequality \ref{eq:let_Daria_win}, we get
    \begin{align*}
        a &\geq \AThr'(t-1,n-1,\piSuff)\left[1 -  \frac{\AThr'(t-1,n-1,\piSuff) - \AThr'(t-1,n,\piSuff)}{k+\AThr'(t-1,n-1,\piSuff)} \right]d \\
          &= \AThr'(t-1,n-1,\piSuff)\left[\frac{k+ \AThr'(t-1,n,\piSuff)}{k+\AThr'(t-1,n-1,\piSuff)} \right]d
    \end{align*}

    Thus, by recursively defining
    \begin{align}
    \AThr'(t,n,\pi) = 
    \AThr'(t,n,(\mathbb{S}, \piSuff)) := \AThr'(t-1,n-1,\piSuff)\cdot \left[\frac{k+\AThr'(t-1,n,\piSuff)}{k+\AThr'(t-1,n-1,\piSuff)} \right]
    \label{eq:AThr'_special_round}
    \end{align}
    We have shown, by definition of $\AThr$, that 
    \[
    \AThr(t,n,\pi,d) = \AThr'(t,n,\pi)\cdot d,  
    \]
    as desired.
    
    We can do an analogous reasoning for the case that the round at the given state is regular, i.e. $p_1 = \mathbb{R}$.  In this case, the optimal bid by Daria is
    \begin{align}
    b^*_{\mathbb{R}} = \frac{\AThr'(t-1,n-1,\piSuff) - \AThr'(t-1,n,\piSuff)}{1+\AThr'(t-1,n-1,\piSuff)} d,
    \label{eq:b^*_R}
    \end{align}
    and the winning condition for Alice in this case is $\AThr(t,n,\pi,d) = \AThr'(t,n,\pi)\cdot d$, where
    \begin{align}
        \AThr'(t,n,\pi) = 
        \AThr'(t,n,(\mathbb{R}, \piSuff)) := 
        \AThr'(t-1,n-1,\piSuff)\cdot \left[\frac{1+\AThr'(t-1,n,\piSuff)}{1+\AThr'(t-1,n-1,\piSuff)}\right]
        \label{eq:AThr'_regular_round}.
    \end{align}
    We omit the details of the proofs of these analogous claims.

    We now show that the dependence of $\AThr'(t,n,\pi)$ on $\pi$ only amounts to the number of special (equivalently, regular) rounds in $\pi$ and not to the internal ordering between the special and regular rounds.  That is, we show the existence of the function $\A$ from the theorem statement.  Note that this would immediately imply parts 1 and 2 in the theorem statement, by substituting $\AThr'$ with $\A$ in Inequalities \ref{eq:AThr'_special_round} and \ref{eq:AThr'_regular_round}, respectively.  Therefore, this is the only remaining piece of the theorem proof.

    To this end, let $\pi = (p_1,\ldots,p_t)$, $\pi'=(p'_1,\ldots,p'_t)$ be two different binary strings in $\{\mathbb{S},\mathbb{R}\}^t$, with the same number of special rounds.
    We need to prove that $\AThr'(t,n,\pi) = \AThr'(t,n,\pi')$.
    
    We first claim that we can assume, without loss of generality, that $\pi$ and $\pi'$ only differ by a single adjacent switch.  
    In other words, we can assume that there exists some index $1\leq i \leq t-1$ such $p_j = p'_j$ for every $j \in \{1,\ldots,t\} \setminus\{i,i+1\}$, and additionally, one of the following holds:
    \begin{itemize}
        \item $(p_i,p_{i+1}) = (\mathbb{R},\mathbb{S})$ and $(p'_i,p'_{i+1}) = (\mathbb{S},\mathbb{R})$, or
        \item $(p_i,p_{i+1}) = (\mathbb{S},\mathbb{R})$ and $(p'_i,p'_{i+1}) = (\mathbb{R},\mathbb{S})$.
    \end{itemize}
    The reason why this is without loss of generality is that between $\pi$ and $\pi'$ there is a sequence of intermediary strings such that each two adjacent strings differ by only a single adjacent switch.

    To this end, recall that Equations \ref{eq:AThr'_special_round} and \ref{eq:AThr'_regular_round} provide a recursive manner to evaluate $\AThr'$.  Note that in the right-hand side of both of these, all recursive calls to $\AThr'$ are applied to the suffix binary string obtained from the input string by removing the first element.  
    The former equation is applied if the first element of the string is $\mathbb{S}$, and the latter is applied if it is $\mathbb{R}$.
    
    Thus ---
    assuming that $\pi$ and $\pi'$ differ by a single adjacent switch and that the divergence between them occurs at index $1\leq i \leq t-1$ ---
    if we take $\AThr'(t,n,\pi)$ and $\AThr'(t,n,\pi')$ and evaluate them this way, then the first $i-1$ levels of recursive applications of Equations \ref{eq:AThr'_special_round} and \ref{eq:AThr'_regular_round} are identical in both, producing the same expression structure.  The divergence between the two evaluations happens at the i'th application.
    Therefore, 
    denoting $\pi'' = (p_{i+2},\ldots, p_t)$, 
    it is enough to show that
    \begin{lemma}
    For all $2\leq t$ and $2 \leq n \leq t$, we have
    \[\AThr'(t,n,(\mathbb{S},\mathbb{R},\pi'')) = \AThr'(t,n,(\mathbb{R},\mathbb{S},\pi''))\]
    \end{lemma}

    We conclude the proof of the theorem with the proof of the lemma.
    \begin{proof}
    We split the proof to cases.
    First, if $n=t$ then the claim holds by part 3-(d).  We can therefore assume for the rest of the proof that $n<t$.
    
    We now add the assumption that $3\leq n$ (so that $1\leq n-2$).
    By applying Equation \ref{eq:AThr'_special_round} followed by three ``parallel'' applications of Equation \ref{eq:AThr'_regular_round}, we get
    \begin{align}
    &\AThr'(t,n,(\mathbb{S},\mathbb{R},\pi'')) = \nonumber\\
    &\AThr'(t-1,n-1,(\mathbb{R}, \pi''))\cdot \left[\frac{k+\AThr'(t-1,n,(\mathbb{R}, \pi''))}{k+\AThr'(t-1,n-1,(\mathbb{R}, \pi''))} \right] = \nonumber\\
    &
    \AThr'(t-2,n-2,\pi'')\cdot \left[\frac{1+\AThr'(t-2,n-1,\pi'')}{1+\AThr'(t-2,n-2,\pi'')}\right]
    \cdot \left[\frac{k+
    \AThr'(t-2,n-1,\pi'')\cdot \left[\frac{1+\AThr'(t-2,n,\pi'')}{1+\AThr'(t-2,n-1,\pi'')}\right]
    }{k+
    \AThr'(t-2,n-2,\pi'')\cdot \left[\frac{1+\AThr'(t-2,n-1,\pi'')}{1+\AThr'(t-2,n-2,\pi'')}\right]
    } \right]
    \label{eq:S_then_R}
    \end{align}

    Note that the above transitions do not hold when $n=t$, which is why we had to treat that case separately.  Now, by applying Equation \ref{eq:AThr'_regular_round} followed by three ``parallel'' applications of Equation \ref{eq:AThr'_special_round}, we similarly get

    \begin{align}
    &\AThr'(t,n,(\mathbb{R},\mathbb{S},\pi'')) = \nonumber \\
    &\AThr'(t-1,n-1,(\mathbb{S},\pi''))\cdot \left[\frac{1+\AThr'(t-1,n,(\mathbb{S},\pi''))}{1+\AThr'(t-1,n-1,(\mathbb{S},\pi''))}\right] = \nonumber \\
    &
    \AThr'(t-2,n-2,\pi'')\cdot \left[\frac{k+\AThr'(t-2,n-1,\pi'')}{k+\AThr'(t-2,n-2,\pi'')} \right]
    \cdot \left[\frac{1+
    \AThr'(t-2,n-1,\pi'')\cdot \left[\frac{k+\AThr'(t-2,n,\pi'')}{k+\AThr'(t-2,n-1,\pi'')} \right]
    }{1+
    \AThr'(t-2,n-2,\pi'')\cdot \left[\frac{k+\AThr'(t-2,n-1,\pi'')}{k+\AThr'(t-2,n-2,\pi'')} \right]
    }\right] 
    \label{eq:R_then_S}
    \end{align}

    We have verified that the expressions \ref{eq:S_then_R} and \ref{eq:R_then_S} are indeed equal, using the \texttt{Wolfram Mathematica} software.

    Finally, we assume that $n=2$. 
    {The reason we need to treat this differently is that in this case we have $n-2=0$, but $\AThr'(t-2,0,\pi'')$ is not defined.}
    Let $s$ be the number of special rounds in $\pi''$.  Then, similarly to how we obtained Equation \ref{eq:S_then_R}, we get 
    \begin{align}
    &\AThr'(t,2,(\mathbb{S},\mathbb{R},\pi'')) = \nonumber\\
    &\AThr'(t-1,1,(\mathbb{R}, \pi''))\cdot \left[\frac{k+\AThr'(t-1,2,(\mathbb{R}, \pi''))}{k+\AThr'(t-1,1,(\mathbb{R}, \pi''))} \right] = \nonumber\\
    &
    \left[ks+(t-1-s)\right]
    \cdot \left[\frac{k+
    \AThr'(t-2,1,\pi'')\cdot \left[\frac{1+\AThr'(t-2,2,\pi'')}{1+\AThr'(t-2,1,\pi'')}\right]
    }{k+
    ks+(t-1-s)
    } \right] = \nonumber \\
    &
    \left[ks+(t-1-s)\right]
    \cdot \left[\frac{k+
    \left[
    ks+t-2-s
    \right]
    \cdot \left[\frac{1+\AThr'(t-2,2,\pi'')}{1+
    ks+t-2-s
    }\right]
    }{k+
    ks+(t-1-s)
    } \right],
    \label{eq:S_then_R_n_is_2}
    \end{align}
    where here we have also used part 3-(a) from the theorem statement various times.
    Now, similarly to how we obtained Equation \ref{eq:R_then_S}, we get
   
    \begin{align}
    &\AThr'(t,2,(\mathbb{R},\mathbb{S},\pi'')) = \nonumber \\
    &\AThr'(t-1,1,(\mathbb{S},\pi''))\cdot \left[\frac{1+\AThr'(t-1,2,(\mathbb{S},\pi''))}{1+\AThr'(t-1,1,(\mathbb{S},\pi''))}\right] = \nonumber \\
    &
    \left[
    k(s+1) + (t-1-(s+1))
    \right]
    \cdot \left[\frac{1+
    \AThr'(t-2,1, \pi'')\cdot \left[\frac{k+\AThr'(t-2,2,\pi'')}{k+\AThr'(t-2,1, \pi'')} \right]
    }{1+
    k(s+1) + (t-1-(s+1))
    }\right] = \nonumber \\
    &
    \left[
    k(s+1) + (t-1-(s+1))
    \right]
    \cdot \left[\frac{1+
    \left[
    ks+t-2-s
    \right]
    \cdot \left[\frac{k+\AThr'(t-2,2,\pi'')}{k+
    ks+t-2-s
    } \right]
    }{1+
    k(s+1) + (t-1-(s+1))
    }\right]
    \label{eq:R_then_S_n_is_2}
    \end{align}
        As we did for the proof of the case $3\leq n$, we have also verified using \texttt{Wolfram Mathematica} that the expressions \ref{eq:S_then_R_n_is_2} and \ref{eq:R_then_S_n_is_2} are equal.  This concludes the proof of the lemma.
    
    \end{proof}
\end{proof}


\begin{observation}
   The proof of Theorem \ref{thm:optimal_specials} shows that for the sake of winning the game, Daria's optimal bid is almost always a fraction of her current budget and does not depend on Alice's budget (if Daria knows that Alice's budget is low enough she could win by bidding less than her optimal prescribed bid, but that would only affect her leftover budget and not her prospects of winning).  The only case where the optimal bid does depend on Alice's budget is when $t=n$, i.e. when Daria has to win all remaining rounds.  In that case Alice might as well bid her entire budget $a$, and Daria's best strategy is to bid $a+\epsilon$ in regular rounds and $a/k +\epsilon$ in special rounds.
\end{observation}

\begin{observation}
    The recursive nature of $\A$ calls for a dynamic programming algorithm to compute it.  Such an algorithm would compute $\A(t,n,s)$ in $O(tn)$ time, which is practical for realistic values of $t,n,s$. In fact, the runtime of the algorithm can be reduced to $O(sn)$ or $O((t-s)n)$, due to the boundary conditions for $s=0$ and $t=s$ derived in Theorem~\ref{thm:optimal_specials}, respectively. See section~\ref{sec:eval} for numerical details.
\end{observation}
    
\begin{remark}
    Since only the number of remaining special rounds affects $\AThr$ and not their interleaving with the regular rounds, we update the notation of a state from $(t,n,\pi,d,a)$ to $(t,n,s,d,a)$ from now on.
\end{remark}

Unfortunately, we could not find a closed-form formula for $\A$ (despite much effort) and we suspect that such a formula does not exist.
However, we can still analyze the asymptotic behavior of $\A$ for some regime of the parameters.
To this end,
we first derive lower and upper bounds on $\A$.

\begin{theorem}\label{optimal_vs_trivials}
    Let $1\leq t$, $1\leq n \leq t$ and $0\leq s\leq t$, such that $n-1 \leq s$.  Then we have 
    \[\frac{t+s(k-1) - k(n-1)}{n}\leq \A(t,n,s) \leq \frac{t+s(k-1)-(n-1)}{n}. \]
\end{theorem}

\begin{proof}
    Let $t,n,s$ be as above, and consider the state $(t,n,s,d,a)$, where $d$ and $a$ are the budgets of Daria and Alice, respectively.
    Suppose that, starting from this state and until the end of the game, Daria constantly bids the same bid $b=\frac{d}{n}$.
    Then, the optimal strategy of Alice is to let Daria win $n-1$ special rounds and win all of the remaining $s-n+1$ special and $t-s$ many regular rounds. Therefore, if Alice's budget $a$ is lower than $\frac{d}{n}k(s-n+1)+\frac{d}{n}(t-s)$ then she loses the game. This gives the required lower bound on $\A(t,n,s)$.

    Now we switch to proving the upper bound.
    Consider the following strategy by Alice: she bids $d/n$ in every regular round and $kd/n$ in every special round. To outbid Alice, Daria must spend more than
    $d/n$ in every round she wins. Since the total budget of Daria is $d$, she can win at most $n-1$ rounds which is not enough. Also note that Alice can allow Daria to win $n-1$ rounds (which in the worst case are all regular rounds). Thus Alice wins if her budget is 
    at least $\left(\frac{t-s-(n-1)}{n}+\frac{sk}{n}\right)d$.  This proves the claim.
\end{proof}

%

Let $TS(t,n,s):=\frac{t-s}{n} + \frac{k(s-n+1)}{n}$ be the lower bound from Theorem~\ref{optimal_vs_trivials}.
By the ``squeeze'' (or ``sandwich'') theorem in calculus, the bounds from the theorem
 immediately imply the following.

\begin{corollary}
\label{cor:asymptotic}
\[
\lim_{\frac{kn}{t+(k-1)s}\rightarrow 0 } \frac{TS(t,n,s)}{\A(t,n,s)} = 1
\]

\end{corollary}

That is, if $kn$ is low enough compared to $t+(k-1)s$, the simple strategy by Daria of \emph{constantly} bidding $d/n$ performs asymptotically as well as the optimal strategy does, when starting from the state $(t,n,s,d,a)$.
We remark that for realistic values of $t,n,s,k$, the expression $\frac{kn}{t+(k-1)s}$ is indeed very small.  See section \ref{sec:eval} for more details.

Furthermore, this corollary has implications for a probabilistic variant of $\mathcal{G}^1_k$, which we describe in the next section.

\subsection{A Probabilistic Variant of $\mathcal{G}^1_k$}

    Consider the game $\mathcal{G}^1_{k,p}$ with parameters $k,p, T,N,D,A$, where $k,T,N,D,A$ are as in $\mathcal{G}^1_k$, and $p \in [0,1]$ is the (known) probability that a given round is special.
    The different round types are i.i.d., and both players observe the type of any round only when that round starts.
    Apart from that, this game is conducted exactly like $\mathcal{G}^1_k$.
    
    Since $T$ is very large in practice, the probability that the actual (random) number of special rounds is very different than the (known) expectation is exponentially low due to concentration bounds on the Bernoulli random variables.  Thus, the simple strategy of constantly bidding $D/N$, as suggested by Corollary \ref{cor:asymptotic}, is still approximately optimal with very high probability.
    In this section we make this intuition precise.
    
    Similarly to $\mathcal{G}^1_k$, the state prior to any round is determined by the tuple $(t,n,d,a)$.  Like the parameter $k$, we also omit $p$ from the state description since it is fixed throughout the game's execution.  Note that the expected number of remaining special rounds starting from state $(t,n,d,a)$ is $tp$.

    \begin{proposition}
    \label{prop:probabilistic_variant}
        Consider the game $\mathcal{G}^1_{k,p}$.
        Let $\epsilon > 0$ and $\delta > 0$.  Let $t,n$ such that $\left|\frac{TS(t,n,\lfloor tp-\delta\rfloor )}{\A(t,n,\lfloor tp-\delta \rfloor)} - 1 \right| < \epsilon$ (such $t,n$ exist by Corollary \ref{cor:asymptotic}).  Then $\Pr_s\left[\left|\frac{TS(t,n,s)}{\A(t,n,s)} - 1 \right| < \epsilon \right] \geq 1- 2e^{-\frac{2\delta^2}{t}} $.
    \end{proposition}

    \begin{proof}
        First, we observe that for any $x \in \left[\lfloor tp-\delta \rfloor, \lfloor tp + \delta \rfloor \right]$, we have
        \[
        \frac{kn}{t+(k-1)x} \leq \frac{kn}{t+(k-1)(\lfloor tp-\delta \rfloor)}.
        \]

        Therefore, for any such $x$, if $s = x$ then Corollary \ref{cor:asymptotic} implies that $\left|\frac{TS(t,n,s)}{\A(t,n,s)} - 1 \right| < \epsilon$.
        We now use Hoeffding's inequality to establish that $s \in \left[\lfloor tp-\delta \rfloor, \lfloor tp + \delta \rfloor \right] $ with probability at least $1- 2e^{-\frac{2\delta^2}{t}}$, and we are done.
    \end{proof}

\section{$\mathcal{G}^m$ Game}
\label{sec:g_m_game}

The game $\mathcal{G}^m$ is similar to $\mathcal{G}^1$ and is also parameterized by $T,N,D,A$.  
However, unlike $\mathcal{G}^1$, in the game $\mathcal{G}^m$ there are $m$ block builders in each round which form an inclusion committee.
In order for Daria to win a round, i.e. have her transaction be included in a block, at least one of the builders has to include Daria's transaction in that builder's local ``inclusion list''.
Builders are myopic: they choose a strategy that maximizes their expected utility in the current round. 

We can model this interaction by specifying payments Daria and Alice offer builders depending on the number of builders that include Daria's transaction. 
{We consider anonymous payment rules, i.e. payment rules that do not depend on builders' identities.  In particular, we assume}
Daria pays $B^t$ in total to the $t$ builders who include her transaction, {for any $0 < t\leq m$} (we denote $B^0 := 0$), 
and
Alice pays $C^l$ in total to the $l$ builders who do not include Daria's transaction.
{Both $B^t$ and $C^l$ are equally divided between the includers and the non-includers, respectively.}

We consider two natural ways to model competition between Daria and Alice in this setting.
In the first approach, we look at the \emph{budget-balanced mechanism}, where Daria proposes a total payment of $B_i$ in round $i$, independent of the number of builders that include her transaction. All builders that include her transaction share the payment $B_i$ equally. That is, $B_i^t=B_i$ is constant for $0 < t \leq m$.
Alice's response is to offer a fixed payment $c_i$ to anyone who does not include Daria's transaction, that is $C_i^l = lc_i$.

The second approach builds upon the idea of~\cite{censorship_resistance}.
Daria offers a conditional payment contract of the following form: 
in each round $i$, if exactly one builder includes Daria's transaction, the builder is paid some large amount $B_i$. If more than one builder includes her transaction, all including builders are paid some small amount $b_i < B_i$. Builders that do not include her transaction get paid $0$. That is, $B_i^{t}=tb_i$ for $t>1$ and $B_i^t=B_i$ for $t=1$. As before, Alice's response is to offer a fixed payment $c_i$ to anyone who does not include Daria's transaction, that is $C_i^l = lc_i$.
 
We analyze $\mathcal{G}^m$ under these two mechanisms in the next two sections.

\subsection{{Analysis of the Budget-Balanced Mechanism}}
\label{sec:budget_balanced_analysis}

We start by
analyzing the behavior of block builders.
{In particular, we will be analyzing the \emph{symmetric} Nash equilibria in this setting, i.e., equilibria in which builders play the same strategy.}


Consider some round $i$.
First, if $B_i \leq c_i$,
{implying in particular that each builder gets paid more for not including as opposed to including Daria's transaction,}
then there is a single symmetric equilibrium where all builders do not include, and Alice wins with probability 1.
{
Next, if $B_i > m \cdot c_i$, then similarly there is a single equilibrium where all builders include Daria's transaction, and Daria wins the round with probability 1.
Finally,
if $c_i < B_i < m \cdot C_i$,
}
then there is no pure strategy equilibrium of the game $\mathcal{G}^m$ in which all block builders play the same pure strategy.

Therefore, we look for a symmetric totally-mixed Nash equilibrium solution in which there is some probability $p_i$, such that each builder includes Daria's transaction with probability $p_i$ and does not include the transaction with probability $1-p_i$. To pin down $p_i$, we need to solve the indifference condition that the expected utility for some builder obtained from inclusion equals the expected utility from non-inclusion: $E[inclusion]=E[no-inclusion]$. The condition is equivalent to: 

\begin{equation}\label{indifference_balanced}
    \sum_{j=0}^{m-1}{m-1 \choose j}p_i^j (1-p_i)^{m-1-j}\frac{B_i}{j+1} = c_i.
\end{equation}

Simplifying the condition using tricks on binomial coefficients gives the equivalent condition:

\begin{equation}\label{indifference_balanced_simple}
    B_i(1-(1-p_i)^m) = m p_i c_i.
\end{equation} 

The probability that Alice wins round $i$ is equal to the probability that no block builder includes Daria's transaction: 
\[
P_i^A := (1-p_i)^m,
\]
and the probability that Daria wins round $i$ is 
\[
P_i^D := 1 - P_i^A.
\]
For the rest of the analysis we assume that in each round the builders play according to the corresponding symmetric equilibrium.


We now derive a bound on Alice's budget that is equal to $m$ times the bound from Theorem~\ref{thm:one_builder}, under which she can win the entire game $\mathcal{G}^m$ with certainty.

\begin{theorem}\label{condition_alice_k}
    If $\frac{(T-N+1)m}{N}D\leq A$, then Alice has a strategy with which she wins with certainty. 
\end{theorem}

\begin{proof}
Alice applies a similar strategy to the one from the proof of Theorem~\ref{thm:one_builder}: she always sets $c_i = \frac{D}{N}$. 
Note that if $B_i\leq \frac{D}{N}$, then Alice wins round $i$ with probability one, and she pays $\frac{mD}{N}$ in that round. 
Thus, if $\frac{(T-N+1)m}{N}D\leq A$, then Alice has enough budget to win $T-N+1$ such rounds.
However, if $B_i>\frac{D}{N}$, then Daria pays more than $\frac{D}{N}$ if she wins that round.
Thus, she cannot win $N$ rounds as that would exceed her budget. This concludes the proof of the theorem. 
\end{proof}

To complement Theorem \ref{condition_alice_k}, we derive an almost matching bound to $m$ times the bound in Theorem~\ref{thm:one_builder}, under which Daria can win with high probability.

\begin{theorem}\label{condition_daria_k}
If $A<\frac{(T-4N+1)m}{N}D$, then Daria has a strategy with which she wins with probability at least $1-\frac{1}{2^{\Theta(N)}}$. 
\end{theorem}

\begin{proof}
Let us assume that $A<\frac{(T-4N+1)m}{N}D$, and
consider the following strategy by Daria:  in each round $i$, she sets $B_i = B = \frac{D}{N}$. Note that with this strategy Daria can never lose because she ran out of budget, since she pays at most $\frac{D}{N}$ when she wins a round
and she needs to win $N$ rounds in order to win the entire game.

To counteract this strategy, in each round Alice either sets a high value for $c_i$, which gives a low probability for Daria to win the round but costs Alice a lot. Or, Alice sets a low value for $c_i$, which costs little but implies a high probability for Daria to win. 

{
Recall that we denote the probability that Daria wins round $i$ by $P^D_i$.
We denote by $S = \{i \mid c_i < B\}$ the set of rounds where Daria has a positive probability of winning (see the beginning of section \ref{sec:budget_balanced_analysis}).
\begin{lemma}
\label{lem:P_D_i_bound}
    If $c_i < B$, then $ 1-(\frac{c_i}{B})^{\frac{m}{m-1}} \leq P^D_i $.
\end{lemma}

\begin{proof}
    If $m\cdot c_i \leq B$, then, as observed at the beginning of section \ref{sec:budget_balanced_analysis}, $P^D_i=1$ and we are done.  We therefore assume that we are in the regime $c_i < B < m\cdot c_i$ in which there is only a single symmetric totally-mixed Nash equilibrium describing the builders' behavior.
    
    Let $p_i$ be the corresponding probability that a builder includes Daria's transaction.  Then we need to show that $1-\left(\frac{c_i}{B}\right)^{\frac{m}{m-1}} \leq 1-(1-p_i)^m$,
    which is equivalent to $(1-p_i)^{m-1}\leq \frac{c_i}{B}$. 
    
    From~\eqref{indifference_balanced_simple}, we have that $\frac{c_i}{B}=\frac{1-(1-p_i)^m}{mp_i}.$ Therefore, by plugging in, we need to show 
    $
    (1-p_i)^{m-1} \leq \frac{1-(1-p_i)^m}{mp_i},
    $
    which is equivalent to 
    \[
    mp_i(1-p_i)^{m-1}+(1-p_i)^m \leq 1.
    \]
    This indeed holds since the left-hand side is 
    the probability that either a single builder includes Daria's transaction or no builder does.
\end{proof}
}
Let $X_i$ be a Bernoulli random variable, that is equal to $1$ if Daria wins round $i$ and $0$ if Daria loses it. Then, $\Pr[X_i=1]=P^D_i$. 
{
Let $X:=\sum_{i\in S} X_i$ and $\mu:=\mathbb{E}(X)=\sum_{i\in S} P^D_i$.


For $i \in S$,
}
Let $\epsilon_i$ be such that $c_i = (1-\epsilon_i)B$. 
Then, 
\[
(1-p_i)^m\leq \left(\frac{c_i}{B}\right)^{\frac{m}{m-1}} = (1-\epsilon_i)^{\frac{m}{m-1}}\leq 1-\epsilon_i,
\]
where the first inequality holds by Lemma \ref{lem:P_D_i_bound} and
the last inequality holds since $1-\epsilon_i < 1$. Thus, the probability Daria wins, $P_i^D$ is at least $\epsilon_i$.

Consider two cases: First, 
{
$\sum_
{i\in S}
\epsilon_i \geq 3N$,
}
which implies that $\mu\geq 3N$. 
We apply a Chernoff type bound:

\begin{equation}\label{chernoff}
\Pr[X\leq (1-\delta)\mu]\leq e^{-\frac{\mu \delta^2}{2}},   
\end{equation}

which holds for any $0<\delta<1$. Then, the probability that Daria loses is at most $\Pr[X<N]$ which is equal to $\Pr[X < (1-\frac{2}{3})3N]$ by taking $\delta=\frac{2}{3}$, and is at most $e^{-\frac{3N\cdot 4/9}{2}}=e^{-\frac{2N}{3}}$, from inequality~\eqref{chernoff}.

In the second case, we have
{
$\sum_
{i \in S}
\epsilon_i<3N$. 
}
We claim that Daria wins the entire game with certainty in this case.  To this end, let us assume towards contradiction that Alice wins the entire game for some instantiation of the builders' behavior, 
and denote $W = \{i \mid \textrm{Alice won round $i$}\}$.  Then, by definition, we have $\left|W\right| \geq T-N+1$.
Note that
Alice pays $mc_i$ if she wins any round $i$. Therefore, 
{
if Alice wins the entire game,
then she spends at least 
\begin{align*}
\sum_{i \in W}mc_i &= \sum_{i \notin S}mc_i + \sum_{i \in S \cap W} mc_i \\
 &\geq \sum_{i \notin S}mB + \sum_{i \in S \cap W} m(1-\epsilon_i)B \\
                                            &= \sum_{i \notin S}mB + \sum_{i \in S \cap W} mB - mB \sum_{i \in S \cap W} \epsilon_i \\
                                            &\geq \sum_{i \in W } mB - mB \sum_{i \in S} \epsilon_i \\
                                            &> mB(T-N+1) - mB \cdot 3N \\
                                            &= mB(T-4N+1),
\end{align*}
where the first equality holds since for any $i\notin S$, Alice wins round $i$ with probability 1.
This calculation
gives a lower bound on her budget $A$. 
}
Plugging in $B=\frac{D}{N}$ in the last expression above gives us
$A>\frac{(T-4N+1)m}{N}D,$
contradicting the main assumption of the theorem.
This concludes the proof.
\end{proof}

\begin{remark}
    The proofs of Theorems \ref{condition_alice_k} and \ref{condition_daria_k} specify winning strategies for the corresponding player, which have two appealing properties.  First, they do not depend on the other player's actions.  Second, they are very simple to calculate.
\end{remark}  

\subsection{{Analyis of the Conditional Payment Mechanism}}
Next, we consider the conditional payment mechanism, and we again start by analyzing the builders' behavior in some round $i$.
Similarly to the previous section,
if $B_i \leq c_i$, then there is a single symmetric equilibrium where all builders do not include Daria's transaction, i.e. Alice wins with probability 1.
If $c_i \leq b_i$, then there is a single symmetric equilibrium where all builders include Daria's transaction, i.e. Daria wins with probability 1.
If $b_i < c_i < B_i$, then there is no pure symmetric equilibrium.
let $p_i$ denote the probability that an individual block builder includes Daria's transaction. The symmetric totally-mixed Nash equilibrium indifference condition for a builder is:   

\begin{equation}
    (1-p_i)^{m-1}B_i + (1-(1-p_i)^{m-1})b_i = c_i,
\end{equation}

implying the solution $$1-p_i=\left(\frac{c_i-b_i}{B_i-b_i}\right)^{\frac{1}{m-1}}.$$

The probability that Alice wins a round is equal to the probability that no block builder includes Daria's transaction: $$P_i^A:=(1-p_i)^m=\left(\frac{c_i-b_i}{B_i-b_i}\right)^{\frac{m}{m-1}},$$ 
and Alice's expected expenditure is $E_i^A=(1-p_i)mc_i.$ The probability that Daria wins a round is $P_i^D:=1-P_i^A$ and Daria's expected expenditure is 
$$E_i^D:=mp_i(1-p_i)^{m-1}B+\sum_{t=2}^{m}{m \choose t}p_i^t(1-p_i)^{m-t}tb_i.$$

We can obtain the same result in this setting as in Theorem~\ref{condition_daria_k}.
\begin{proposition}\label{condition_daria_k_assymetric}
If $A<\frac{D(T-4N+1)m}{N}$, then Daria wins with probability at least $1-\frac{1}{2^{\Theta(n)}}$. 
\end{proposition}

\begin{proof}
    {
    Let us assume that $A<\frac{D(T-4N+1)m}{N}$, and consider Daria's strategy which sets $b_i = b =0$ and $B_i = B = \frac{D}{N}$ in each round $i$.
    In every round $i$ for which $c_i < B$ we have $p_i = 1 - \left(\frac{c_i}{B}\right)^{\frac{1}{m-1}}$.
    The proof now proceeds exactly as the proof of Theorem \ref{condition_daria_k}, starting from the point after Lemma \ref{lem:P_D_i_bound}.
    }
\end{proof}

A result similar to Theorem~\ref{condition_alice_k} holds in the conditional payment mechanism case as well, but the proof gets too complicated and less intuitive.

\section{Evaluation}\label{sec:eval}
In this section, we discuss the practical implications of our results from previous sections given real world values for the parameters that define the games $\mathcal{G}^1$, $\mathcal{G}_k^1$, and $\mathcal{G}^m$, and we focus on the case of BoLD \cite{bold}, the new dispute resolution protocol for Arbitrum, which has been deployed on Ethereum mainnet on February 12th, 2025.

The one week challenge period in BoLD corresponds to about $T=50000$ blocks/rounds, since a new block on Ethereum is created every $12$ seconds.
As discussed in the introduction, it is estimated that around $2\%$ of the proposers use default software to decide which transactions get included in their blocks.
In other words, around
$s=1000$ rounds are special.

The block gas limit is $30$M gas units at the time of writing, but there is a strong willingness to raise it to $36M$ \cite{pumpTheGas}. 
A single BoLD transaction by the defender takes up about $0.5$M gas units.
Thus, in order to censor a defender's transaction in a special round,
the adversary would need to submit at least $29.5$M gas units worth of transactions, each of which is specified with a priority fee per gas unit that is higher than the priority fee per gas unit specified by the defender.

Therefore, we estimate the parameter $k$ for the special rounds of Section~\ref{sec:g_1_k_game} to be about $60$, and soon to be increased to about $72$.
The maximum number of transactions that an honest party needs to post in a BoLD challenge is around $n=60$. 
Let us assume that the attacker can steal all bridged assets from the optimistic rollup, if she manages to confirm the wrong claim about the rollup state. In the case of Arbitrum, that would amount to around $\$10$B worth of assets, which gives an estimate for $\A$.
To guarantee security against censorship, 
Theorem~\ref{thm:one_builder} suggests that the defender's budget $D$ should be at least $\$12$M, in the absence of special rounds.

For the special round variant, the value $\A(50000, 60, 1000)$ is approximately equal to $1786$. Therefore, Theorem~\ref{thm:optimal_specials} implies that $D$ should be at least $\$5.6$M in order to be secure against censorship.
 
For these $t,n,s$,
the lower bound from Theorem~\ref{optimal_vs_trivials}, $TS(t,n,s)$ for $s\geq n-1$ approximates $\A(t,n,s)$ by $98.5\%$.  i.e. $\frac{TS}{\A}\approx 0.985$.
However, for some parameter instantiations, this bound is not a good approximation. For example, when $s$ and $n$ are close, $\frac{TS}{\A}$ can reach $35.1\%$.  An example of this is $t=214$, $s=n=57$ and $k=25$. Then, $\A(t,n,s)=9.1$ whereas $TS(t,n,s)=3.19$. The upper bound derived from Theorem~\ref{optimal_vs_trivials} is equal to $26.8$. This means that none of the bounds come close to the optimal value, $\A$. The reason for this is that the value $kn$ is significant compared to $t$ and $s$. 

With $m$ block builders, the required value for $D$ is simply lowered by a factor of $m$. A reasonable value for multiple proposers or FOCIL committee members is around $20$, hence the required defender's budget is in the order of sub million.

\section{Conclusion}
\label{sec:conclusion}
This paper studies economic censorship on blockchains, in scenarios that require multiple sequential transactions to be posted by the honest party, a prominent example being fraud proof systems for optimistic rollups. In three different settings, we derive bounds on the ratio between the budgets of the attacking and defending parties under which one of them can guarantee victory in the censorship game.

All the analysis in the paper is based on priority fees which resemble auction bids. However, Ethereum also uses base fees to charge for transactions, in order to handle congestion. Since the base fee grows exponentially when consecutive blocks are filled, it is not a sustainable strategy for the attacker to outbid the defender for too long in consecutive special rounds.  However,
some mixture of priority fee outbidding and base fee manipulation might prove to be a useful tactic for the attacker.
This is unclear and we suggest to incorporate base fee dynamics into the model for future work.

Another potential future direction is to model a setting in which the attacking party can initiate multiple challenge games in parallel, each of them at a (high) cost. Arbitrum BoLD is an example of a challenge protocol that allows this. Doing this increases the number of transactions the defending party might need to post sequentially, but this comes at a cost for the attacker, who would need to place a large deposit for each game.

\paragraph{Acknowledgement}
We would like to thank Mario M. Alvarez, Julian Ma, Max Resnick and Terence Tsao for helpful discussions and feedback.

\bibliographystyle{plain}
\bibliography{bibliography}

\end{document}